\documentclass[11pt]{article}
\usepackage{cite,hyperref}
\usepackage{amsthm}
\usepackage{graphicx,floatrow,wrapfig}

\usepackage[margin=1in]{geometry}

\newtheorem{theorem}{Theorem}
\newtheorem{lemma}[theorem]{Lemma}

\begin{document}

\title{The Graphs of Planar Soap Bubbles}
\author{David Eppstein\thanks{Computer Science Department, University of California, Irvine}}
\date{}
\maketitle

\begin{abstract}
We characterize the graphs formed by two-dimensional soap bubbles as being exactly the 3-regular bridgeless planar multigraphs. Our characterization combines a local characterization of soap bubble graphs in terms of the curvatures of arcs meeting at common vertices, a proof that this characterization remains invariant under M\"obius transformations, an application of M\"obius invariance to prove bridgelessness, and a M\"obius-invariant power diagram of circles previously developed by the author for its applications in graph drawing.
\end{abstract}

\thispagestyle{empty}
\newpage
\setcounter{page}{1}
\pagestyle{plain}

\section{Introduction}

\begin{figure}[t]
\centering
\includegraphics[height=2.75in]{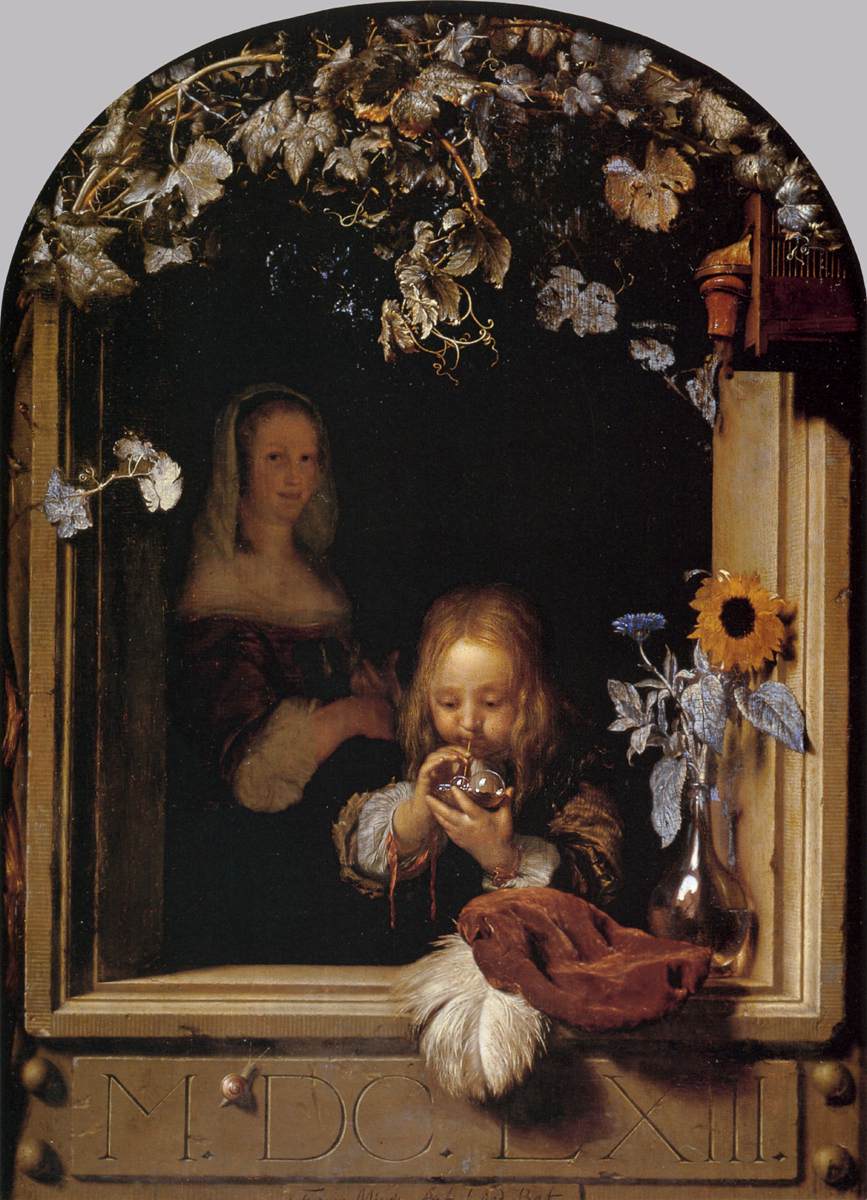}
\qquad\qquad
\includegraphics[height=2.75in]{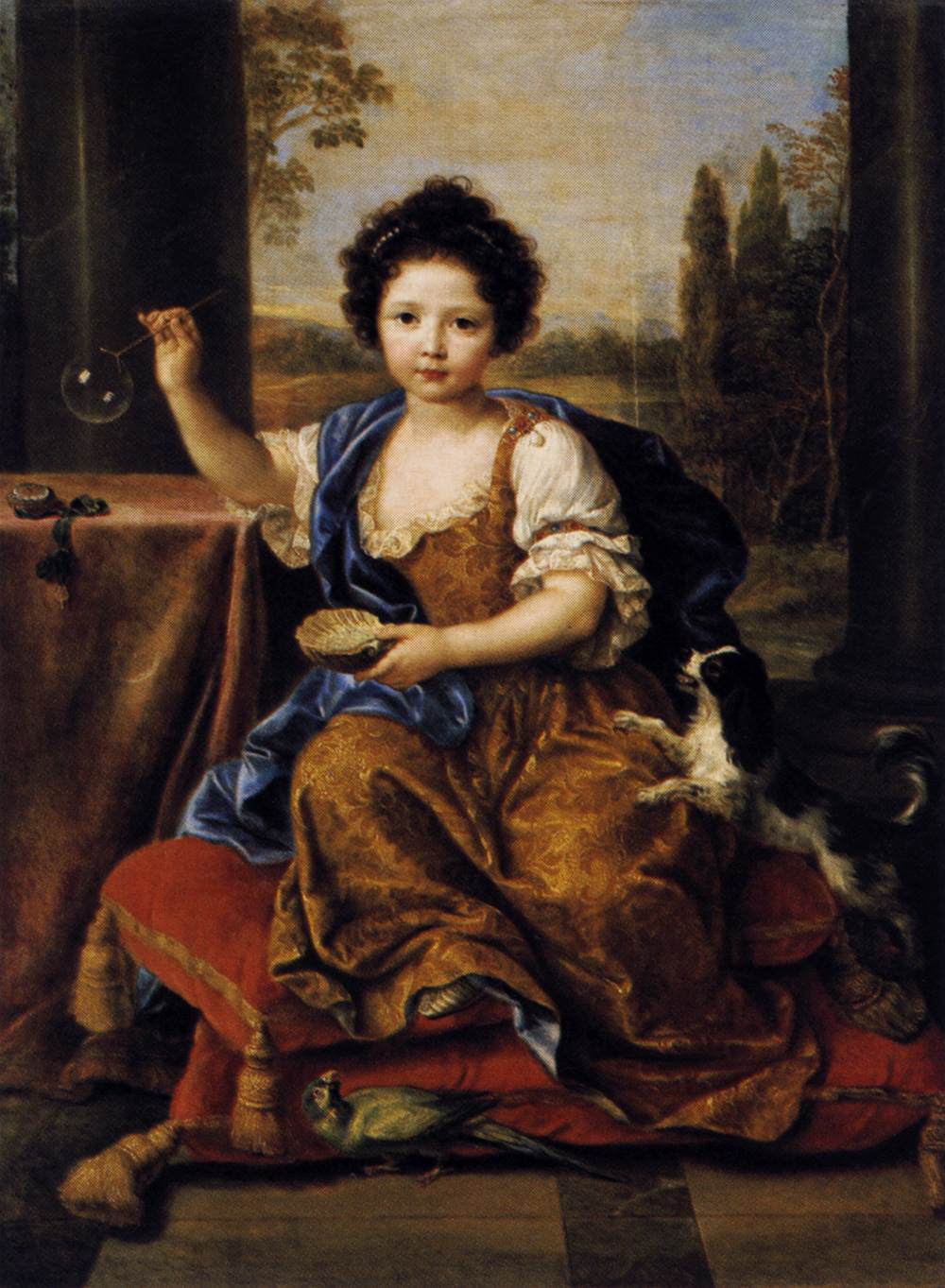}
\caption{Left: {\it Boy Blowing Bubbles} (1663), by 
Frans van Mieris.
Right: {\it Girl Blowing Soap Bubbles} (1674), by Pierre Mignard. Both images public domain from Wikimedia Commons.}
\label{fig:art}
\end{figure}

Blowing soap bubbles has long been a popular pastime for children and a subject of art (Figure~\ref{fig:art}), but the mathematics of soap bubbles has also been a subject of study and some mystery for many years, and the geometry and combinatorial structure of soap bubbles are not well understood. The Kelvin conjecture, stating that a foam with truncated-octahedron cells has the minimum total surface area among all soap bubble foams with equal-volume cells, was posed in the 19th century by Lord Kelvin and finally solved negatively in 1993 by the discovery of the Weaire--Phelan structure, a more parsimonious foam with two different cell shapes~\cite{Tho-PM-87,WeaPhe-PML-94}; proof that this newer structure is optimal remains elusive. A related problem, the minimum-area cap for the ends of hexagonal honeycomb cells, also remains open~\cite{Hal-NAMS-00}. It is conjectured that minimum-area four-bubble clusters, for given amounts of gas in each bubble, form stereographic projections of a four-dimensional simplex, with each 2-face of the simplex projecting to a spherical surface in three-dimensional space, but this has not been proven, and there are clusters of six bubbles for which the numerically-computed surfaces separating the bubbles are not spherical~\cite{Sul-FE-97}. Even as simple-sounding a problem as the double bubble conjecture (proving that the optimal structure for clusters of two bubbles has three spherical patches sharing a common circle) took many years to solve~\cite{HasSch-AM-00,HutMorRit-AM-02}

In this paper we study the special case of soap bubbles formed between glass plates, set close enough to each other that each bubble surface spans the entire gap between the plates. Effectively, these bubbles are confined to a two dimensional plane rather than three dimensional space. Such a system is familiar through the ``soap bubble computer'', in which pegs are placed between the two plates at specified locations; a soap film that connects the pegs forms a minimal (though generally not minimum) Steiner network, providing a heuristic and approximate physical solution to an NP-complete problem~\cite{Aar-SN-05,BerGra-SA-89,DutKhaRoy-AJP-10,Hof-TMT-79,Ise-AS-76}. However, we consider here bubbles supported only by the two glass plates, without  additional pegs (in equilibrium, assuming perfectly impermeable bubbles, but not necessarily globally minimizing the total bubble length). When restricted in this way, bubble clusters may be described combinatorially by planar graphs, with vertices where three bubbles meet and edges for the interfaces between pairs of bubbles. With infinitely many bubbles of equal areas, the optimal shape for the bubbles is a hexagonal tiling of the plane~\cite{Hal-NAMS-00,Hal-DCG-01}, and similar tilings arise also for finitely many equal bubbles~\cite{CoxMorGra-12}, but other conditions give rise to more irregular arrangements of bubbles (Figure~\ref{fig:2dfoam})
The question we consider is, with finitely many bubbles of variable areas, which planar graphs are possible?

\begin{wrapfigure}[20]{r}{.4\textwidth}
  \centering
  \includegraphics[width=2.25in]{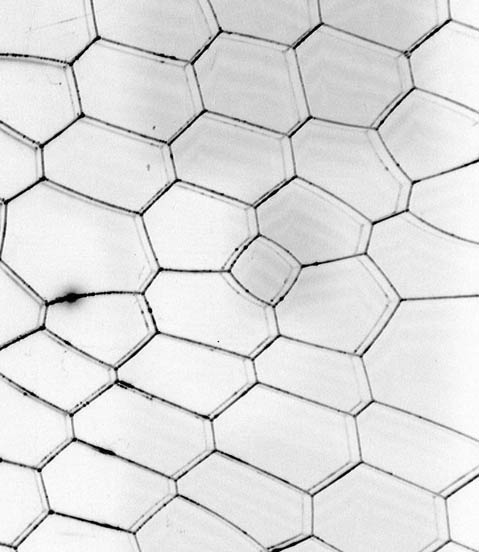}
  \caption{Soap bubbles between two glass plates. Public domain negative image by Klaus-Dieter Keller \href{http://commons.wikimedia.org/wiki/File:2-dimensional_foam_(colors_inverted).jpg}{from Wikimedia commons}.}
  \label{fig:2dfoam}
\end{wrapfigure}

We provide a complete answer to this question, by showing that the graphs of planar  soap bubble clusters are exactly the bridgeless three-regular planar graphs. To do this we use two sets of long-known facts about soap bubbles: Plateau's laws, describing the local geometry of the surfaces and junctions in a soap bubble cluster, and the Young--Laplace equation, which relates the curvature of soap bubble surfaces to the pressures of the gas within the bubbles. Based on these facts we prove two additional properties of planar soap bubbles clusters: a characterization of them in terms of purely local conditions on the curvatures of the arcs meeting at each vertex, and a lemma stating that a M\"obius transformation of a planar soap bubble cluster is itself also a planar soap bubble cluster. The fact that a soap bubble cluster must form a bridgeless graph follows easily by combining this M\"obius transformation property with the Young--Laplace equation.
We show that bridgelessness and 3-regularity are not just necessary but also sufficient  by using two of the three steps in a recent method of the author for finding circular-arc drawings of planar graphs~\cite{Epp-lombardi}. The first step applies to 3-vertex-connected graphs, and applies a novel M\"obius-invariant power diagram of disks (defined using three-dimensional hyperbolic geometry) to a family of disks constructed from the planar dual of the given graph by the Koebe--Andreev--Thurston circle packing theorem~\cite{Ste-ICP-05}. The second step generalizes from 3-vertex-connected graphs to 2-vertex-connected graphs by means of SPQR trees~\cite{DiBTam-ICALP-90}. (The third step of this graph drawing algorithm, generalizing from 2-vertex-connected graphs to arbitrary graphs, does not work for soap bubbles.) Because of our use of power diagrams, our result can be seen as a validation of Sullivan's suggestion that ``we might look for foams as relaxations of Voronoi decompositions''~\cite{Sul-FE-97}.

Soap bubbles are closely related to \emph{Lombardi drawings} of graphs, drawings in which all edges are represented as circular arcs meeting at equal angles at each vertex~\cite{DunEppGoo-GD-11,DunEppGoo-GD-10,DunEppGoo-JGAA-12,Epp-lombardi}, and as detailed above our characterization uses methods previously used for Lombardi drawing; however, not all planar 3-regular Lombardi drawings represent soap bubbles, as Figure~\ref{fig:valid} (left) shows.
Our result can also be seen as an analogue for soap bubbles of \emph{Steinitz's theorem} that the graphs of three-dimensional convex polyhedra are exactly the 3-vertex-connected planar graphs~\cite{Ste-EMW-22}, and of
several related theorems characterizing the combinatorial structure of geometric objects. Similar graph-theoretic characterizations (with different connectivity conditions) are also known for simply-connected polyhedra with axis-parallel edges~\cite{EppMum-SCG-10} and for subdivisions of rectangles into smaller rectangles~\cite{KozKin-DAC-84,LeiLai-DAC-84,Ung-JLMS-53}, among other structures. Indeed, the Koebe--Andreev--Thurston theorem that we use here can also be viewed as a result in this vein.

\section{The classical conditions}

\begin{figure}[t]
\centering\includegraphics[height=1.75in]{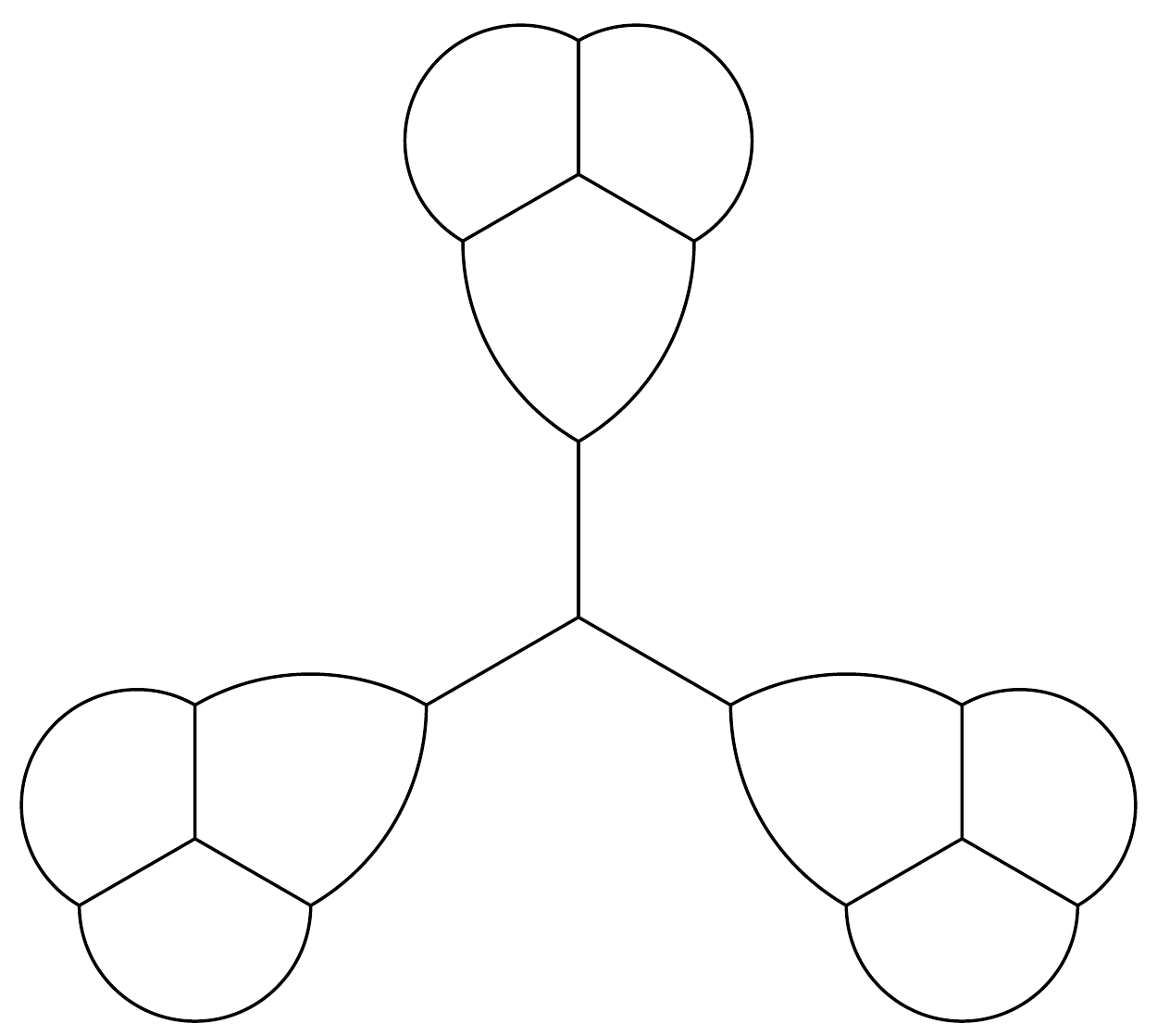}\qquad
\includegraphics[height=1.75in]{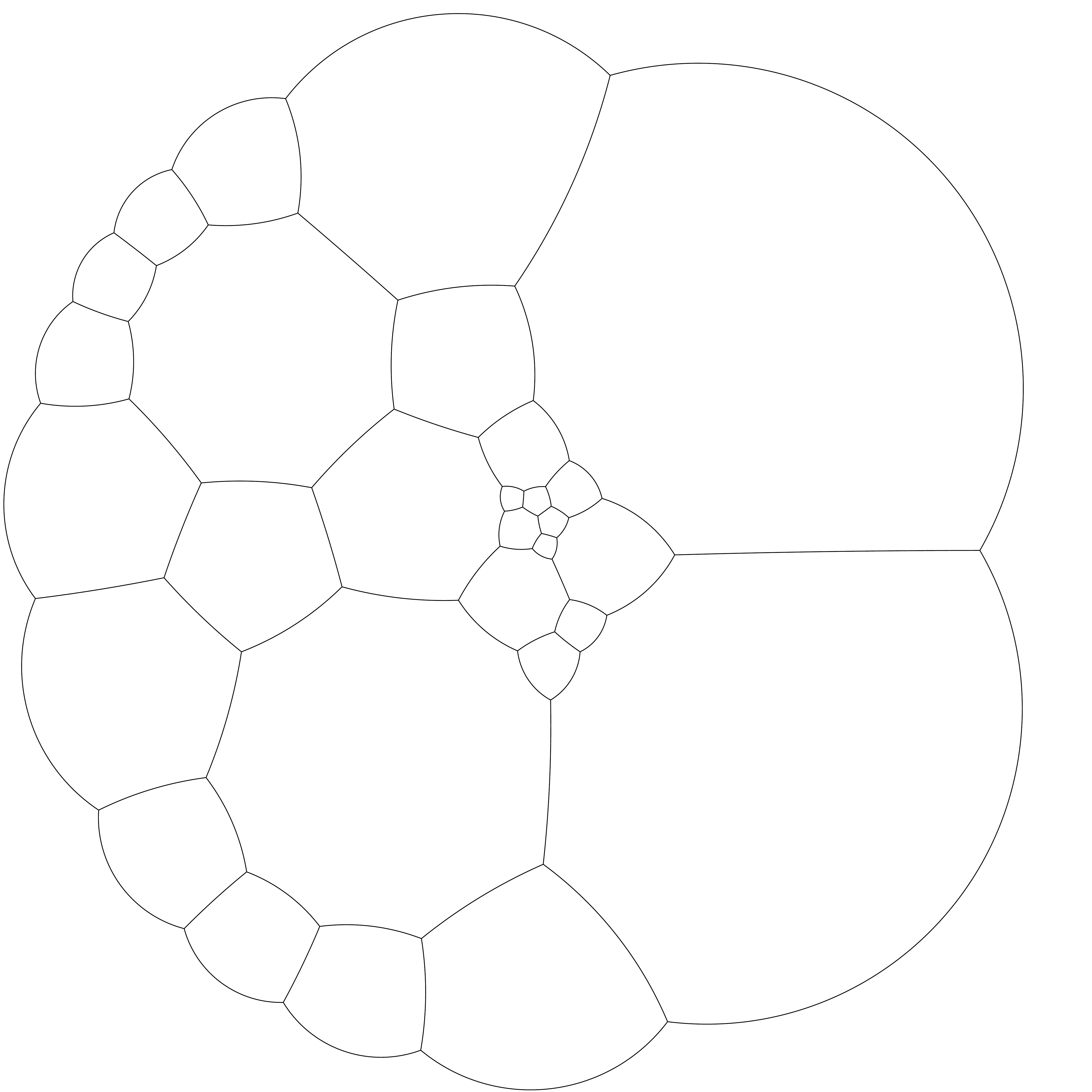}\qquad
\includegraphics[height=1.75in]{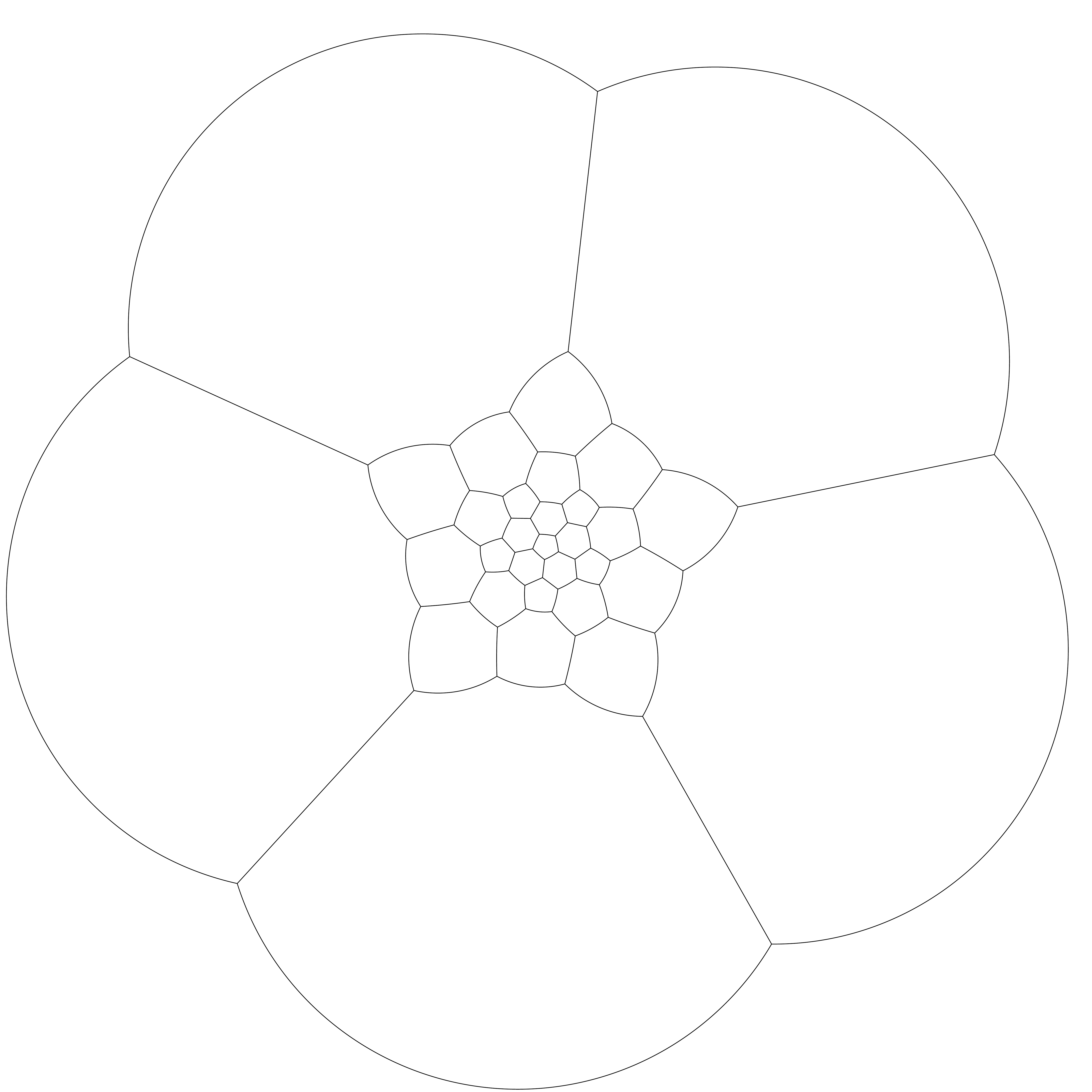}
\caption{Left: A collection of segments and arcs that obeys Plateau's laws, but does not form a  planar soap bubble cluster, because it disobays the Young--Laplace equation (the curvatures do not add to zero at some of its vertices). Three of its segments do not separate different regions, impossible for soap bubbles. Modified from a figure in~\cite{DunEppGoo-JGAA-12}. Center: An irregular planar soap bubble cluster. Right: a soap bubble cluster with the topology of a buckyball.}
\label{fig:valid}
\end{figure}

Two classical results on soap bubbles may be used to describe their geometry. \emph{Plateau's laws}, observed experimentally in the 19th century by Joseph Plateau~\cite{Plateau-1873} and proven rigorously for minimal surfaces in 1976 by Jean Taylor~\cite{Tay-AM-76}, state that in a three-dimensional soap bubble cluster,
\begin{itemize}
\item Each two-dimensional surface has constant \emph{mean curvature}. That is, the two principal curvatures of the surface take the same average value at all points in the surface.
\item At each one-dimensional junction of surfaces, exactly three surfaces meet, and they form dihedral angles of $2\pi/3$ with each other. Such a junction is called a \emph{Plateau border}.
\item At each endpoint of a Plateau border, exactly four Plateau borders and six two-dimensional surfaces meet, and the borders form angles of $\cos^{-1}(-1/3)$ with each other (the same angle that would be formed by two rays from the center of a regular tetrahedron through two of its vertices).
\end{itemize}
The Young--Laplace equation,  formulated in the 19th century by Thomas Young and Pierre-Simon Laplace, states that the mean curvature of a two-dimensional surface in a soap bubble cluster is proportional to the difference in pressure on the two sides of the surface, with a constant of proportionality determined by the surface tension of the fluid forming the soap bubbles~\cite{Ise-LY-78}.

In the case of a soap bubble cluster formed between two flat plates, the bubble walls are perpendicular to the plates, and all cross-sections of the bubbles by a plane parallel to one of the plates are congruent to each other. In this case, the principal curvature in the direction perpendicular to the plates is zero, and we may simplify both laws to describe the planar figure formed by the cross-sections:
\begin{itemize}
\item The figure consists of one-dimensional curves of constant curvature; that is, circular arcs or line segments that do not cross each other.
\item At each endpoint of one of these arcs or segments, exactly three curves meet, and they form angles of $2\pi/3$ with each other.
\item The curvature of any one of the circular arcs (the inverse of its radius) is proportional to the difference in pressure between the bubbles it separates. Bubbles with the same pressure as each other are separated by line segments, with zero curvature.
\end{itemize}

If a system of bubbles and bubble boundaries obeys these laws, the local forces on each vertex or arc caused by pressure and surface tension will sum to zero, so it will be in an equilibrium state (though possibly an unstable one~\cite{ForVaxCox-CSA-07,WeaCoxGra-EPJE-02}). Thus,
we define a planar soap bubble cluster to be a finite collection of curves, obeying the planar restrictions of Plateau's laws and the Young--Laplace equation (for some assignment of pressures to each bubble). Equivalently, in graph-theoretic terminology, it is a planar embedding of a graph with circular-arc edges, forming angles of $2\pi/3$ at each vertex, for which it is possible to find a pressure assignment to the bubbles that is consistent with all of the arc curvatures. As we describe in the next section, it is possible to replace this existential condition (the existence of a consistent pressure assignment) with a simpler statement about the curvatures of the three arcs meeting at each vertex. Examples of planar soap bubble cluster are shown in Figure~\ref{fig:valid}, center and right.

\section{Local characterization}

To state our local characterization of  planar soap bubble clusters, we need a signed variation of curvature. Given a circular arc ending at a point $p$, of curvature $c\ge 0$, we define the \emph{signed curvature} of the arc at $p$ to be $c$ whenever the arc curves clockwise as it leaves $p$, and to be $-c$ whenever it curves counterclockwise. For line segments, of course, the signed curvature is zero.

\begin{lemma}
\label{lem:local}
A finite collection of circular arcs and line segments forms a planar soap bubble cluster if and only if it obeys Plateau's laws and if, at each endpoint of an arc or segment, the sum of the signed curvatures of the three incoming curves is zero.
\end{lemma}

Essentially the same result (for bubbles on a sphere instead of in the plane) can be found in Quinn Maurmann's  appendix to~\cite{Qui-RHUMJ-07}, so we defer the proof of the lemma to an appendix.

\section{M\"obius invariance}

When three three-dimensional bubbles form a cluster in which all surface patches are spherical, it was known to Plateau that the three inner surface patches have centers of curvature that all lie on a single line~\cite{HerAre-PMaPS-91}. Although centers of circles and collinearity of points are not M\"obius-invariant properties, we can restate the same collinearity property for planar bubble clusters in a M\"obius-invariant way, in terms of the crossing points of the circles. This  restatement will allow us to prove that the M\"obius transformation of a planar soap bubble cluster is itself another planar soap bubble cluster.

\begin{lemma}
\label{lem:equivcond}
Let three circular arcs meet at angles of $2\pi/3$ at a point $X$, with centers of curvature $C_i$ for $i\in\{1,2,3\}$, and let $r_i=|XC_i|$ be the radii of each of the three arcs. Then the following three conditions are equivalent:
\begin{enumerate}
\item The sum of the three signed curvatures of the arcs is zero.
\item The three center points $C_i$ are collinear.
\item The three circles with centers $C_i$ and radii $r_i$ have two triple crossing points.
\end{enumerate}
\end{lemma}

\begin{figure}[t]
\centering\includegraphics[scale=0.46]{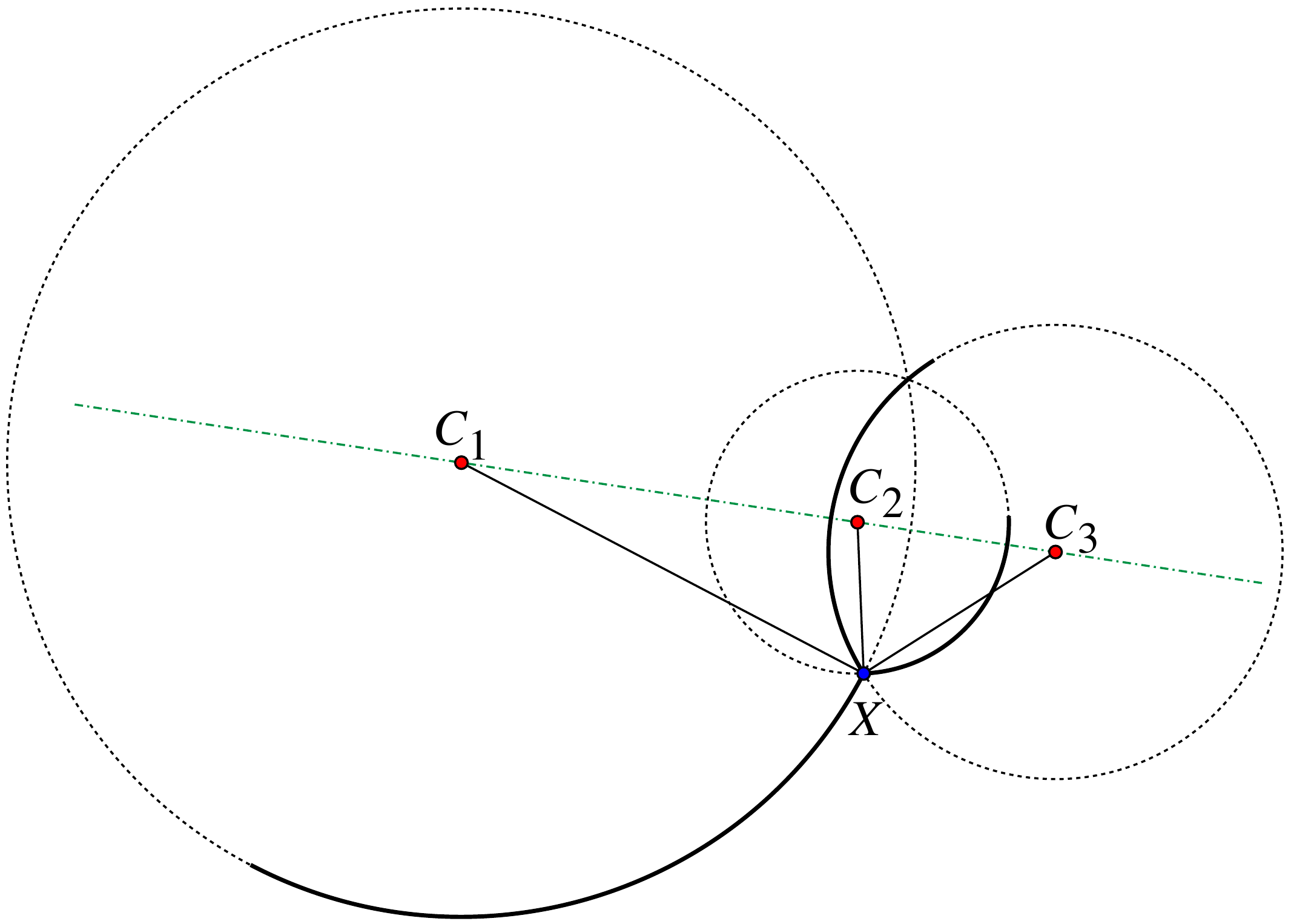}
\caption{Arcs meeting the conditions of Lemma~\ref{lem:equivcond}}
\label{fig:equivcond}
\end{figure}

Again, we defer the proof to an appendix.
This triple crossing property is closely related to, and inspired by, the existence and uniqueness of the standard double bubble, for which see e.g. Proposition 14.1 of~\cite{Mor-GMT-4ed}.
Figure~\ref{fig:equivcond} shows a set of of arcs, circles, centers, and radii meeting the conditions of the lemma.
The lemma applies directly only to circular arcs and not to straight line segments, but the first and third conditions are equivalent more generally: For two arcs and one segment that meet at a point $X$, obeying Plateau's laws and having zero curvature sum, the two arcs must have equal and opposite curvatures and there are again two triple crossings between the two circles and the line containing the three given curves. It is not possible for one arc and two segments to have zero curvature sum, and the case of three line segments can be thought of as having a second triple crossing ``at infinity''.

The \emph{M\"obius transformations} are a family of transformations of the plane, augmented by a single point at infinity; if the points of the plane are represented by complex numbers, each such transformation is either a fractional linear transformation
$$z\mapsto \frac{az+b}{cz+d}$$
(where $a$, $b$, $c$, and $d$ are complex numbers with $ad-bc\ne 0$)
or its complex conjugate. The M\"obius transformation of a circle or line is necessarily another circle or line, and the M\"obius transformations can be characterized as the largest group of transformations that preserves circles in this way. M\"obius transformations are also \emph{conformal}: if two curves meet at an angle $\theta$, their transformed images also form the same angle. The M\"obius transformation of a line segment or circular arc is generally another line segment or circular arc, but there are two additional undesired cases: it could instead be a ray, or a pair of oppositely-oriented rays on a common line.

M\"obius transformations do not preserve physically meaningful quantities of soap bubbles such as their pressures or the total length or area of soap film. They do not even preserve the finiteness of the bubbles: it is possible to find a transformation that takes the unbounded region of the plane to a bounded region and vice versa. Nevertheless, surprisingly, they transform soap bubble clusters into other soap bubble clusters:

\begin{lemma}
\label{lem:mobius}
Let $B$ be a planar soap bubble cluster, and let $\tau$ be a M\"obius transformation with the property that no curve or arc of $B$ is transformed into a ray or double ray. Then $\tau(B)$ is also a planar soap bubble cluster.
\end{lemma}

\begin{proof}
We first consider the case that neither $B$ nor $\tau(B)$ has any straight segments.
By Lemma~\ref{lem:local} and Lemma~\ref{lem:equivcond}, $B$ consists of arcs meeting at angles of $2\pi/3$ such that each three arcs that meet have circles forming two triple crossings.
The meeting angles and the triple crossings of circles are preserved by M\"obius transformation,
so $\tau(B)$ is also a collection of arcs meeting at angles of $2\pi/3$ such that each three arcs that meet have circles forming two triple crossings. By Lemma~\ref{lem:equivcond} again, the signed curvatures at each meeting point of $\tau(B)$ add to zero, and by Lemma~\ref{lem:local} again, $\tau(B)$ is a planar soap bubble cluster.

If $B$ or $\tau(B)$ contains straight segments, the equivalence between zero sums of signed curvature and triple crossing points (counting the point at infinity as a crossing of any two lines) at any endpoint of a segment can be shown even more easily, as detailed in the remarks following Lemma~\ref{lem:equivcond}, and the result follows in the same way.
\end{proof}

Lemmas~\ref{lem:equivcond} and~\ref{lem:mobius} are stated in different terminology as Lemmas 5.1 and 5.2 of~\cite{Wic-PhD-02}. A hint of this M\"obius transformation property of bubble clusters can also be found in Sullivan's conjecture that four-bubble clusters are stereographic projections of a 4-dimensional simplex~\cite{Sul-FE-97}, as any two such projections are related to each other by a three-dimensional M\"obius transformation.

\section{Bridgelessness}

Our graph-theoretic characterization of planar soap bubble clusters asserts that the graphs of these clusters are bridgeless, or equivalently (since they are three-regular graphs) that they are two-vertex-connected. In the language of soap bubbles, this means that all of the arcs and segments in the cluster separate two different bubbles. As we now prove, this is necessarily true for a planar soap bubble cluster.

\begin{lemma}
\label{lem:bridgeless}
In a planar soap bubble cluster, it is not possible for a segment or arc to have the same bubble on both of its sides.
\end{lemma}

\begin{proof}
It is not possible for a circular arc to have the same bubble on both sides, because the pressure would necessarily be the same on both sides and the circularity of the arc would violate the Young--Laplace equation. But then it is also not possible for a line segment to have the same bubble on both sides, because it is straightforward to find a M\"obius transformation that transforms this segment into a curved arc, violating either the Young--Laplace equation or Lemma~\ref{lem:mobius}.
\end{proof}

\section{A M\"obius-invariant power diagram of disks}

In~\cite{Epp-lombardi} we used three-dimensional hyperbolic geometry to devise, from a given set of disjoint disks (or complements of disks) in the plane, a partition of the plane into regions bounded by circular arcs that is invariant under M\"obius transformations, and we used this partition to construct Lombardi drawings of 3-regular planar graphs.  In the same paper we observed that the junction between regions for three mutually tangent disks could be found as the isodynamic point of the triangle formed by the three points of tangency, allowing the calculation of the diagram for this case to avoid hyperbolic geometry, but we did not find a purely two-dimensional description of this diagram for arbitrary disjoint disks and we did not describe how to extend it to disks that might intersect.  Here we describe the same structure as a minimization diagram for a distance function defined between (point,disk) pairs in the plane, and we extend it to disks that may not necessarily be disjoint; that is, we show that it can be interpreted as a form of planar \emph{power diagram}. Note that although this diagram is itself M\"obius-invariant, the distance that it minimizes is not; this situation is much like that for the Delaunay triangulation of points on a sphere, which is M\"obius-invariant even though spherical distance is not~\cite{BerEpp-SODA-03}. We keep our description brief and intuitive, for space reasons and because the additional generalization compared to~\cite{Epp-lombardi} is not necessary for our results on soap bubbles.

\begin{figure}[t]
\centering\includegraphics[height=1.5in]{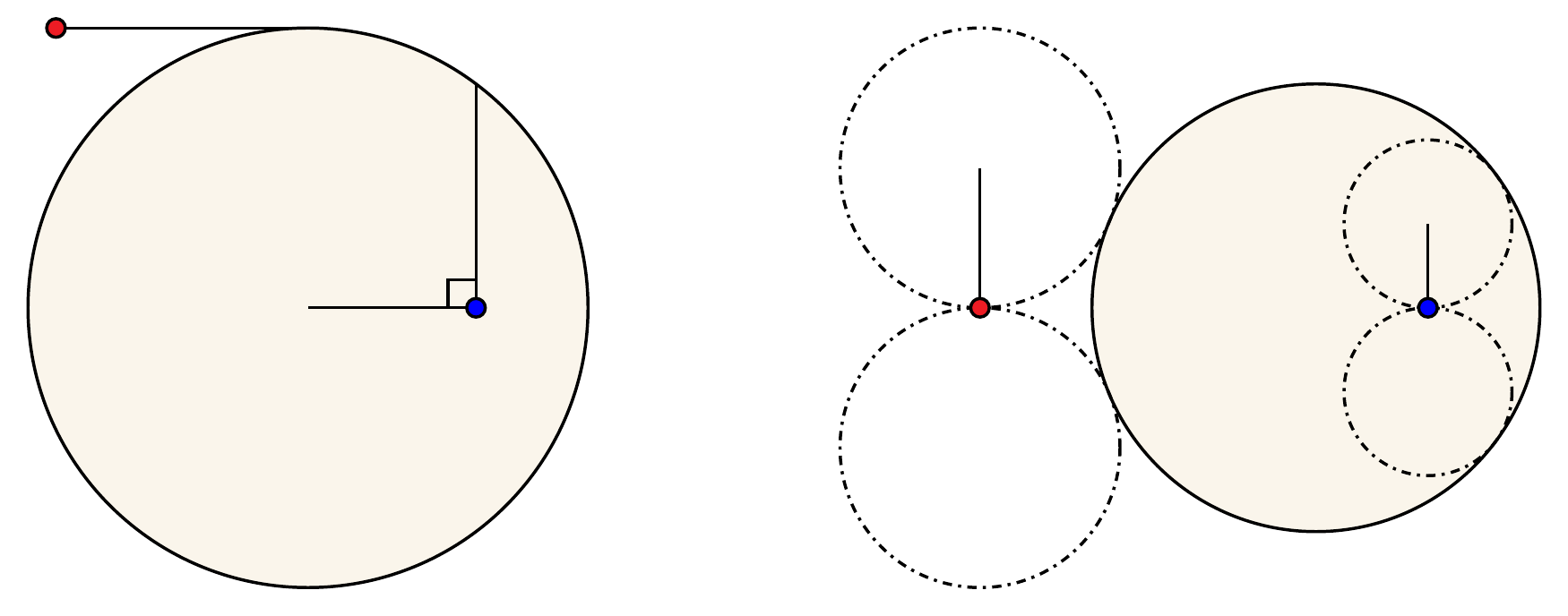}
\caption{Left: The power distance from the red point outside the circle is the length of its tangent segment; the power distance from the blue point inside the circle is half the length of a chord bisected by the point. Right: The radial power distance from the red point outside the disk is the radius of the two congruent circles tangent to each other at the red point and both tangent to the circle; the radial power distance from the blue point inside the disk is the negation of the radius of the two congruent circles.}
\label{fig:radpowdist}
\end{figure}

\begin{wrapfigure}[18]{r}{.35\textwidth}
  \centering
  \includegraphics[width=2.25in]{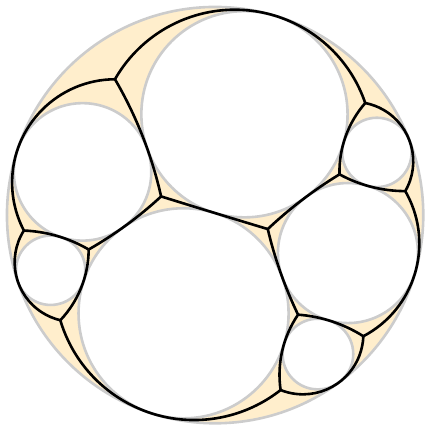}
  \caption{A packing of seven disks and one disk complement, with its M\"obius-invariant power diagram. Modified from a figure in~\cite{Epp-lombardi}.}
  \label{fig:Frucht2}
\end{wrapfigure}

The three-dimensional hyperbolic construction from~\cite{Epp-lombardi} is as follows. We consider the plane in which the disks lie to be the boundary plane of points ``at infinity'' for a three-dimensional Poincar\'e halfspace model of hyperbolic space; the circle bounding each disk $D_i$ is the set of limit points of a plane $P_i$ in this hyperbolic space. We form the three-dimensional Voronoi diagram whose sites are the planes $P_i$, and we extend this diagram from hyperbolic space to the boundary points of the model. The Voronoi diagram is invariant under congruences of hyperbolic space, and these congruences become M\"obius transformations when restricted to the boundary of the model, so the resulting diagram is invariant under M\"obius transformations as desired.

To extend the hyperbolic construction to disks that may intersect, we consider each of the planes $P_i$ in hyperbolic space to be the boundary of a halfspace $H_i$ containing $D_i$. We define the \emph{signed distance} of a point $q$ from $P_i$ to be the positive distance from $q$ to $P_i$ if $q$ is outside of $H_i$, and the negative distance if $q$ is inside of $H_i$. Then, in hyperbolic space, the minimization diagram of these signed distances has bisectors that are hyperbolic planes, so each cell of the minimization diagram is an intersection of halfspaces. When the disks are disjoint, the previous construction (the Voronoi diagram of disjoint hyperbolic planes) is recovered as a special case.
For a point $q$ within hyperbolic space, its nearest neighbor by signed distance can be found (as in the Euclidean case) by considering concentric spheres centered at $q$ and determining either the largest such sphere that is contained in a halfspace $H_i$ or (if no such sphere exists) the smallest such sphere disjoint from all halfspaces. When $q$ is a point at infinity of hyperbolic space, its nearest neighbor can be found in the same way, but using \emph{horospheres} in place of concentric spheres; these are shapes that in the halfspace model of hyperbolic space are modeled as spheres tangent at $q$ to the plane at infinity.

In the Euclidean plane, the classical power distance of a point $q$ to a circle $C$ is either the length of a tangent line segment from $q$ to $C$ (if $q$ is outside $C$) or minus half of the length of a chord bisected by $q$ (if $q$ is inside $C$); see Figure~\ref{fig:radpowdist}, left. If $d$ is the Euclidean distance from $q$ to the center of $C$ and $r$ is the radius of $C$, then by the Pythagorean theorem, in either case, the squared power distance multiplied by the sign of the power distance has the simple algebraic form $d^2-r^2$. The power diagram is the minimization diagram of power distance (or of $d^2-r^2$); it has a polygonal cell for each of a given set of circles within which the power distance to that circle is less than the power distance to any other circle~\cite{Aur-SJC-87}.

Analogously, in the Euclidean plane, given a point $q$ and a disk $D$, define the \emph{radial power distance} from $q$ to $D$ by finding a pair of congruent circles $C_1$ and $C_2$, tangent to each other at $q$ and both tangent to $D$, as shown in Figure~\ref{fig:radpowdist}, right. If $q$ is outside $D$, the radial power distance is the radius of these circles; if $q$ is inside $D$, it is the negation of the radius. Again, by the Pythagorean theorem, the radial power distance has the algebraic form $(d^2-r^2)/2r$. We claim that the M\"obius-invariant diagram of disks defined above by three-dimensional hyperbolic geometry is the same thing as the minimization diagram for radial power distance. To see this, consider a planar cross-section of the three-dimensional hyperbolic space, defined by the plane in the half-space model that passes through $q$ and the center $c$ of $D$ and is perpendicular to the boundary plane. This cross-section contains the closest point of $D$ to $q$; the half-space bounded by $D$ appears in the cross-section as a semicircle congruent to half of $D$, and the largest horosphere containing $q$ and entirely inside or entirely outside of the half-space appears in the cross-section as a circle congruent to $C_1$ and $C_2$. Indeed, the cross-section of the half-space and the horosphere together may be formed from $D$ and $C_1$ or $C_2$ by folding them upwards at right angles from the given Euclidean plane, along line $qc$ (Figure~\ref{fig:radialpower}). From this it can be seen that the minimization diagram of radial power distance prioritizes the circles $C_i$ in exactly the same way that the minimization diagram of hyperbolic signed distance prioritizes the corresponding horospheres, and therefore that the two diagrams coincide.

\begin{figure}[t]
\floatbox[{\capbeside\thisfloatsetup{capbesideposition={right,top},capbesidewidth=8cm}}]{figure}[\FBwidth]
{\caption{Cross-section of a halfspace (blue) and a horosphere externally tangent to it (red) in the half\-space model of hyperbolic space. The plane on which the shadows fall is the plane at infinity of the model, and the black curves show the line $qc$ (where $q$ is the point of tangency of the horosphere to the plane at infinity and $c$ is the center of the disk from which the halfspace was defined) and the two tangent circles in the plane at infinity that are congruent (as Euclidean figures) to the cross-sectional circles.}
\label{fig:radialpower}}
{\includegraphics[height=2in]{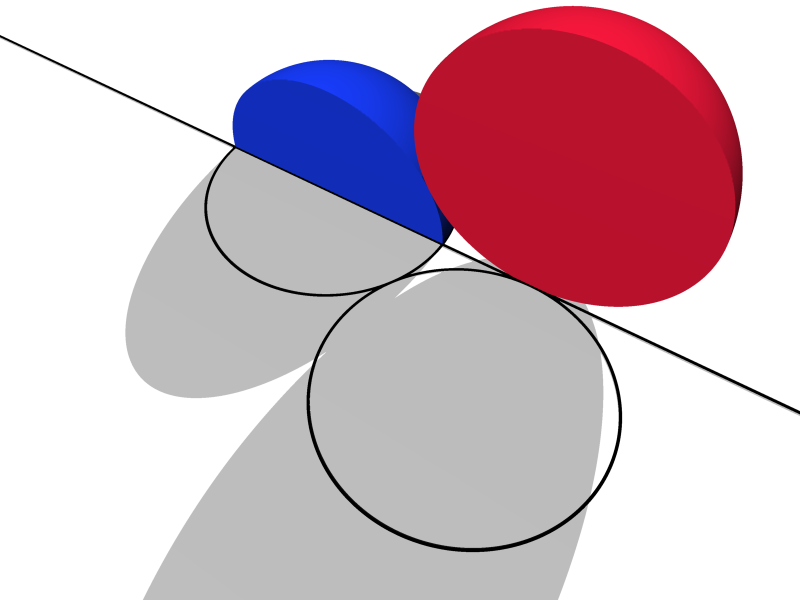}}
\end{figure}

\section{Realizing graphs as soap bubbles}

We are now ready to describe the constructions from~\cite{Epp-lombardi} that we use here to generate a planar soap bubble cluster for any bridgeless 3-regular planar multigraph~$G$.

We consider three different cases, according to the connectivity and existence of multiple edges in the given graph~$G$:
\begin{itemize}
\item If $G$ is 3-vertex-connected multigraph that is not a simple graph, then $G$ must have two vertices and three parallel edges. For, any other connected 3-regular multigraph must have a pair of parallel edges between some pair of vertices $u$ and $v$, each of $u$ and $v$ must also be adjacent to a singleton edge $uu'$ and $vv'$, and the removal of $u'$ and $v'$ would separate $u$ and $v$ from the rest of the graph, violating 3-connectivity. But a graph with three parallel edges is easily realized as a standard double bubble, with one straight line segment and two arcs spanning angles of $4\pi/3$ (Figure~\ref{fig:spqr-glue}, left).

\item If $G$ is 3-vertex-connected and has no multiple edges, its planar dual is a maximal planar graph $G'$. By the Koebe--Thurston--Andreev circle packing theorem~\cite{Ste-ICP-05}, we may find a set $D$ of disks (or complements of disks) with disjoint interiors, corresponding one-for-one with the vertices in $G'$, such that two vertices are adjacent in $G'$ if and only if the corresponding two disks in $D$ are tangent. By performing a M\"obius transformation if necessary we may ensure without loss of generality that $D$ includes one disk complement and that the other elements of $D$ are disks. Let $\Pi$ be the M\"obius-invariant power diagram of $D$, as shown in Figure~\ref{fig:Frucht2}. None of the boundaries between regions in $\Pi$ can cross the disk complement in $D$ and therefore no boundary can be a ray or a double ray. $\Pi$ has a region for each disk of $G'$, so by planar graph duality its vertices and edges form a drawing of $G$~\cite{Epp-lombardi}.

To show that $\Pi$ realizes $G$ as a planar soap bubble cluster,
let $ABC$ be any triangle of the maximal planar dual graph $G'$, and let $p$, $q$, and $r$ be the points of tangency of the corresponding three disks in $D$. An inversion $\xi$ centered at the isodynamic point of triangle $pqr$ (a special type of M\"obius transformation) will take $pqr$ to an equilateral triangle (this property characterizes the isodynamic point and is often used to define it) and therefore causes the transformed images of the three disks to be congruent. By the M\"obius invariance of the power diagram, $\xi(\Pi)$ is the M\"obius-invariant power diagram of $\xi(D)$, and in $\xi(\Pi)$ (by symmetry) the three curves forming the boundaries between the regions for $p$, $q$, and $r$ form straight line segments or rays that meet at angles of $2\pi/3$. These boundary curves lie on three lines with have two triple crossings in the extended plane: one where the segments meet, and one at infinity. Back in $\Pi$, the boundaries between the same three regions are the images under the transformation $\xi^{-1}$ of these three line segments; therefore they are necessarily either line segments or arcs of circles, they again meet at angles of $2\pi/3$, and they again lie on lines or circles that have two triple crossing points. By Lemma~\ref{lem:equivcond}, they also have signed curvatures that sum to zero. We have shown that each three incident curves in $\Pi$ satisfy the conditions of Lemma~\ref{lem:local} and therefore by that lemma $\Pi$ is a planar soap bubble cluster realizing $G$. Figures~\ref{fig:valid} and~\ref{fig:Frucht2} give examples of soap bubble clusters constructed for 3-connected graphs in this way.

\begin{figure}[t]
\centering\includegraphics[height=1.25in]{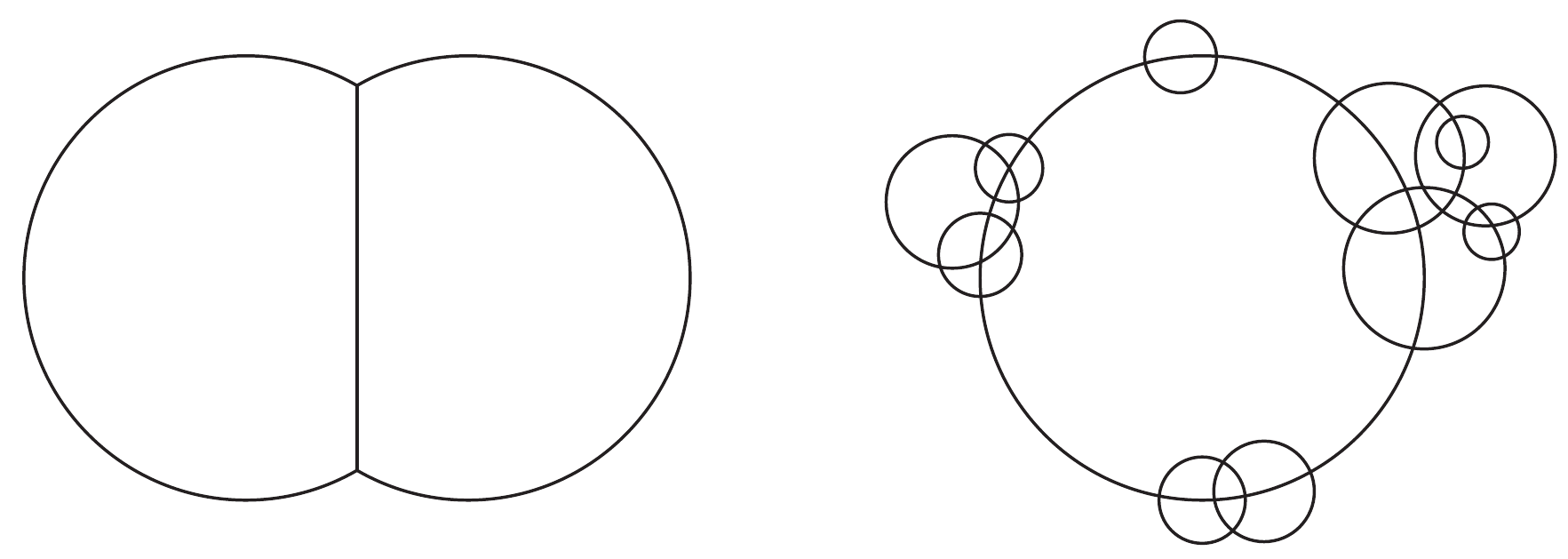}\qquad
\includegraphics[height=1.25in]{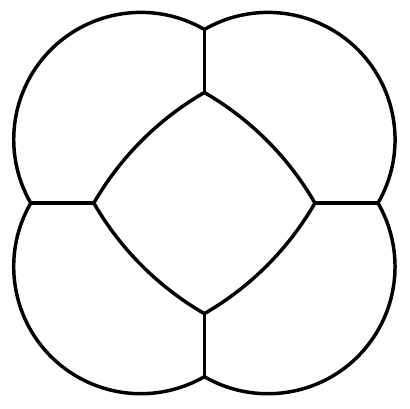}
\caption{Left: double bubble representing a 3-edge multigraph. Center: gluing multiple Lombardi drawings together on an S node with alternating virtual and non-virtual edges (schematic view, not an actual planar soap bubble cluster), modified from a figure in~\cite{Epp-lombardi}. Right: a planar soap bubble cluster that does not come from a circle packing.}
\label{fig:spqr-glue}
\end{figure}

\item Finally, suppose that $G$ is bridgeless but not 3-connected. We may decompose $G$ into 3-connected components, whose connections to each other are described by an SQPR tree (Figure~\ref{fig:spqr})~\cite{DiBTam-ICALP-90,Mac-DMJ-37}. In an SPQR tree, each node represents a 3-connected component of the graph. Each component is a smaller graph that may be a cycle (an S-node), a two-vertex multigraph (a P-node), or a 3-vertex-connected graph (an R-node). Within a component, some edges may be designated as ``virtual'', and each edge of the SPQR tree connects two virtual edges in different components; the original graph $G$ may be recovered by identifying and then deleting pairs of virtual edges. In the SPQR tree of a 3-regular planar graph, each P-node has three edges, each R node is itself 3-regular and planar,
exactly one virtual edge of each linked pair must belong to an S-node, and each S-node must be an even length cycle that alternates between virtual and non-virtual edges~\cite{Poo-3SC-01,EppMum-SCG-10}.

To realize $G$ as a planar soap bubble cluster, we separately realize each $P$-node and $R$-node and then glue together the virtual edges of one $S$-node at a time. To glue together the planar soap bubble clusters linked to an $S$-node, we perform M\"obius transformations of each cluster so that the circular arcs realizing their virtual edges lie on a common circle, and so that the remaining parts of the clusters are shrunken to lie within disjoint disks that meet the circle in the correct order. We then use the arcs from each cluster within the disk that contains it, and the arcs of the shared circle to connect them, as shown schematically in Figure~\ref{fig:spqr-glue} (center). Details of this shrinking and gluing process may be found in~\cite{Epp-lombardi}. Since each triple of incident curves in the resulting drawing is locally a M\"obius transformation of a triple of incident edges in one of the 3-connected components, it again satisfies the conditions of Lemma~\ref{lem:local} and again the whole drawing $\Pi$ is a planar soap bubble cluster.
\end{itemize}

\begin{figure}[t]
\centering\includegraphics[height=1.5in]{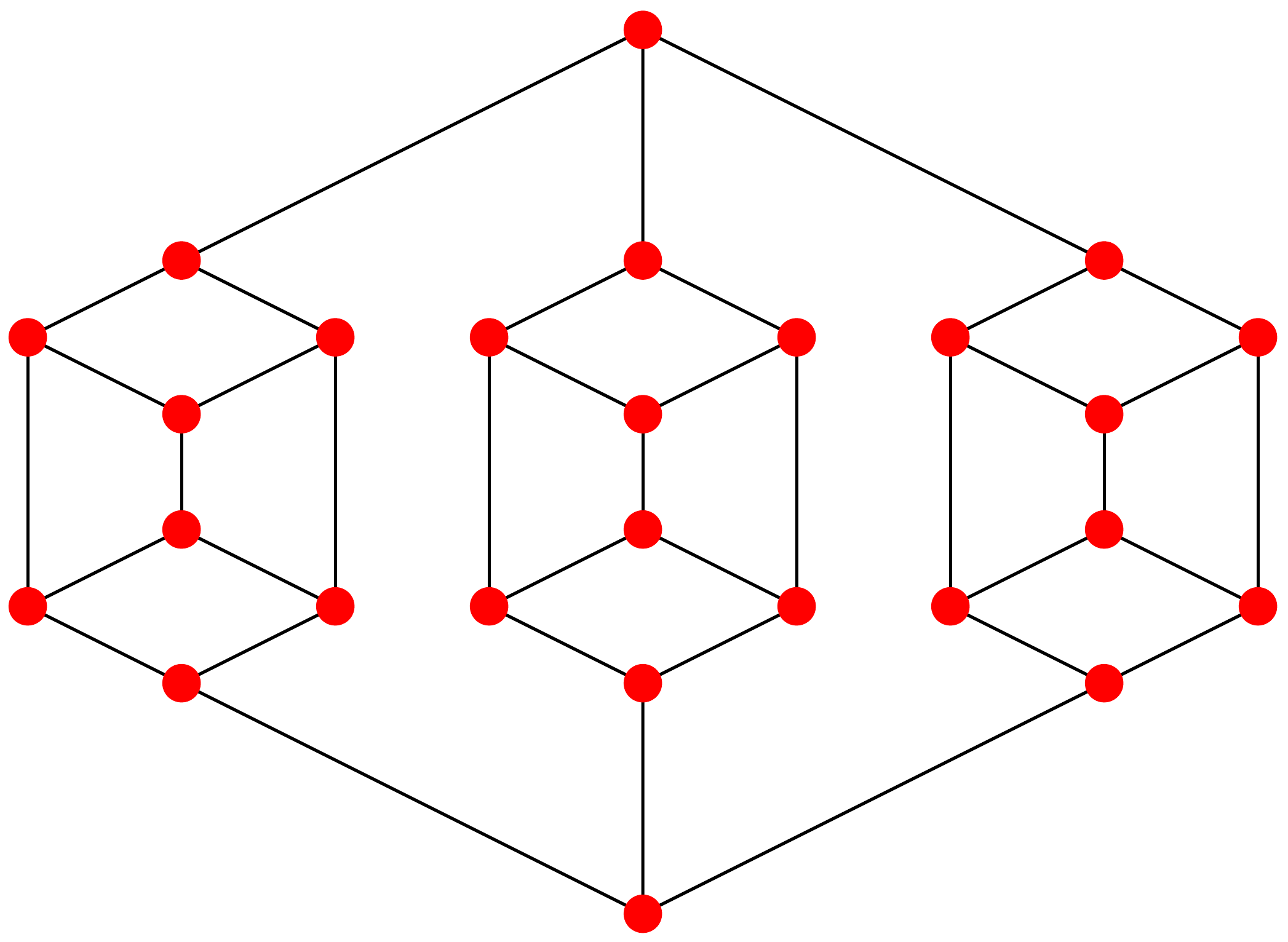}\qquad
\includegraphics[height=1.5in]{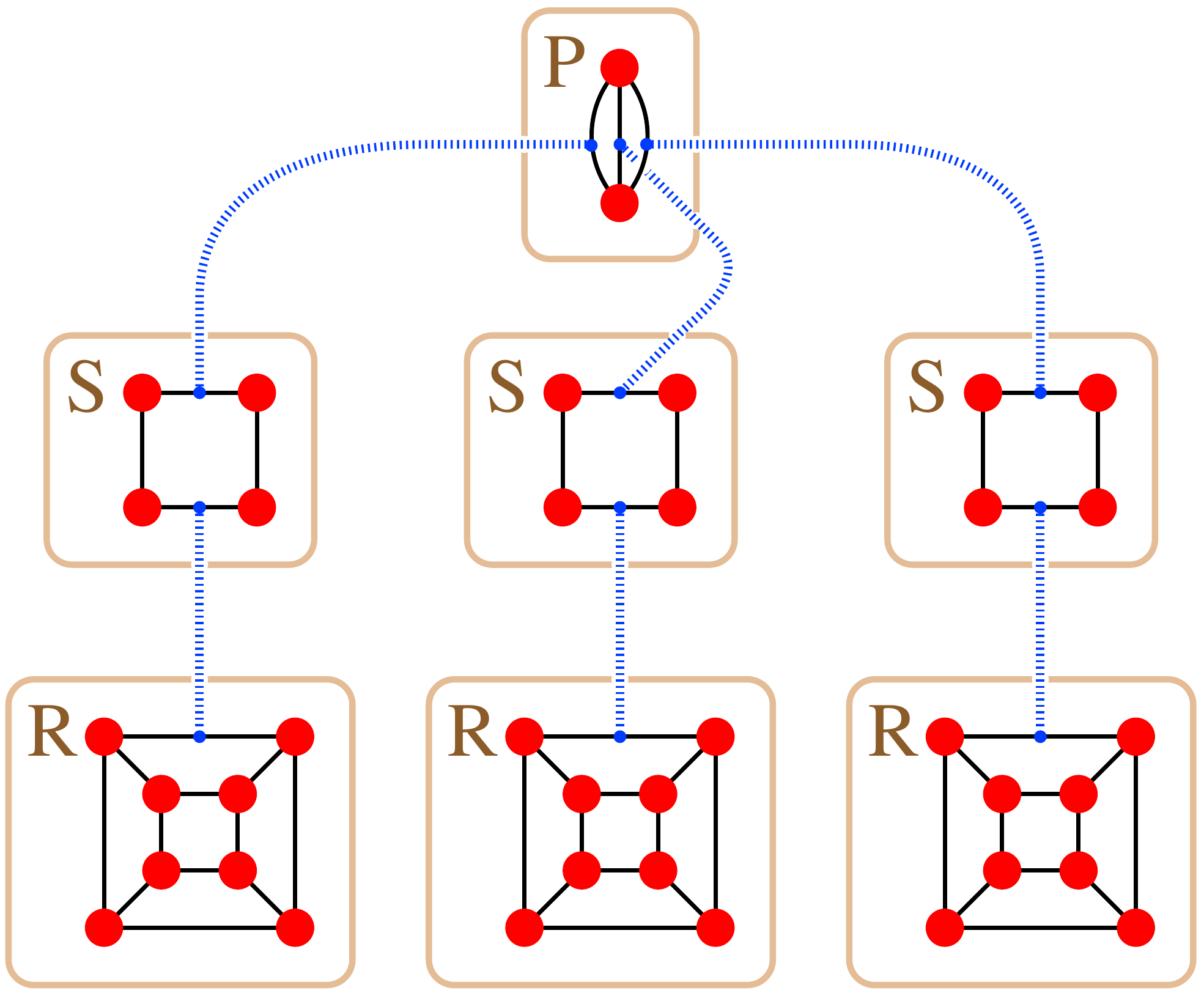}
\caption{A planar 3-regular and bridgeless but not 3-connected graph (left) and its SPQR tree (right). The blue dashed curves link pairs of virtual edges. Figures from~\cite{EppMum-SCG-10}.}
\label{fig:spqr}
\end{figure}

This gives us our main result:

\begin{theorem}
A multigraph $G$ can be realized as the arcs and junctions of a planar soap bubble cluster if and only if $G$ is planar, 3-regular, and bridgeless.
\end{theorem}

\begin{proof}
Planarity and 3-regularity follow from Plateau's laws. The requirement that a planar soap bubble cluster be bridgeless is Lemma~\ref{lem:bridgeless}.
Conversely, the construction outlined in this section shows that every planar 3-regular bridgeless graph has a realization as a planar soap bubble cluster.
\end{proof}

This method can be made into a practical algorithm for realizing graphs as soap bubbles, which we have implemented in the 3-connected case~\cite{Epp-lombardi}, but because it depends on the circle packing theorem it necessarily produces approximate numerical results rather than precise vertex and edge coordinates, and its running time depends on the precision of the results~\cite{ColSte-CGTA-03,Moh-DM-93}.
We note that, although our circle packing method is capable of realizing all graphs of planar soap bubble clusters, it is not capable of generating all valid geometries of planar soap bubble clusters. In particular, the clusters we generate have the property that within each bubble there is a circle tangent to all the arcs of the bubble, but this property does not hold for arbitrary soap bubble clusters; Figures \ref{fig:2dfoam} and~\ref{fig:spqr-glue} (right) provide counterexamples.

\newpage

\subsection*{Acknowledgements}

This work was supported in part by NSF grant
0830403 and by the Office of Naval Research under grant
N00014-08-1-1015.
The author is grateful to Frank Morgan and John Sullivan for helpful comments on a draft of this paper.

{\raggedright
\bibliographystyle{abuser}
\bibliography{bubbles}}

\newpage
\appendix

\section{Proofs of Lemmas}
Recall that Lemma~\ref{lem:local} states that a finite collection of circular arcs and line segments forms a planar soap bubble cluster if and only if it obeys Plateau's laws and if, at each endpoint of an arc or segment, the sum of the signed curvatures of the three incoming curves is zero.
Similar local characterizations were known already (see the appendix of~\cite{Qui-RHUMJ-07}) but
for self-containedness, we give here a proof.

\begin{proof}
In one direction, in a planar soap bubble cluster, let $x$, $y$, and $z$ be the pressures of three bubbles that meet at a point. Algebraically, $(x-y)+(y-z)+(z-x)=0$, and therefore the three signed curvatures which by the Young--Laplace equation are proportional to $x-y$, $y-z$, and $z-x$ also add to zero.

In the other direction, suppose that finitely many circular arcs obey Plateau's laws and have signed curvatures adding to zero at each point where three arcs meet. We must assign pressures to the bubbles formed by the arcs, such that the curvatures obey the Young--Laplace equation. Assign zero pressure to the outside region. For each bounded bubble $B$, choose arbitrarily a curve $c$ that starts in the outside region, ends within $B$, meets the arcs only at proper crossing points interior to an arc, and has only finitely many crossing points. The pressure within $B$ can then be determined as the sum of the pressure differences determined by the curvatures of each crossed arc.

To show that this system of pressures obeys the Young--Laplace equation, consider any two bubbles $B_1$ and $B_2$ separated by an arc $A$, and let $c_1$ and $c_2$ be the curves from the outside region to $B_1$ and $B_2$ by which their pressures were determined. We may link $c_1$ and $c_2$ by another curve that remains entirely within the outside region, forming a single (possibly self-crossing) curve that connects $B_1$ to $B_2$, and we may continuously deform this curve, crossing finitely many triple points of the system of arcs as we do, until it forms a line segment crossing arc $A$. Each time this deformation causes the curve to cross a triple point, the assumption that the curvatures at that point add to zero ensures that the sum of the pressure differences of crossings along the curve remain constant. Before the deformation, this sum was the pressure difference between $B_1$ and $B_2$ in our system of pressures, and after the deformation, it is the pressure difference required between bubbles $B_1$ and $B_2$ for the curvature of $A$ to obey the Young--Laplace equation. Since these two pressure differences are equal, $A$ (and by the same argument every arc) obeys the equation.
\end{proof}

Lemma~\ref{lem:equivcond} states that the following three conditions on three circular arcs meeting at a point $X$, with centers of curvature $C_i$ and radii $r_i$ (as shown in Figure~\ref{fig:equivcond}) are equivalent:
\begin{enumerate}
\itemsep0pt
\item The sum of the three signed curvatures of the arcs is zero.
\item The three center points $C_i$ are collinear.
\item The three circles with centers $C_i$ and radii $r_i$ have two triple crossing points.
\end{enumerate}
Similar statements can be found in the literature (e.g. see Lemma~5.1 of~\cite{Wic-PhD-02}) but again we supply a proof for completeness.

\begin{proof}
We partition the proof into four implications between the three conditions of the lemma.
\begin{description}
\item[{\rm (1)${}\Rightarrow{}$(2):}]~\\
Suppose that the three signed curvatures sum to zero; then one must have a different sign than the other two. Without loss of generality (by permuting the indices and by mirroring the configuration, if necessary) we may assume that the signed curvatures are $1/r_1$, $-1/r_2$, and $1/r_3$.
By assumption, the sum of these three quantities is zero; by multiplying each term of this sum by $r_1r_2r_3$ and rearranging, we obtain the equation $r_1r_2+r_2r_3=r_1r_3$.

Note that the three lines $XC_i$ form angles of $\pi/3$ with each other, because they are perpendicular to the arcs, which meet at angles of $2\pi/3$. The assumption on the signs of the curvatures implies that angles $C_1XC_2$ and $C_2XC_3$ must both equal $\pi/3$, and angle $C_1XC_3$ must equal $2\pi/3$. Now consider the areas of the three triangles $C_1XC_2$, $C_2XC_3$, and $C_1XC_3$. In any triangle, the area can be computed by the side-angle-side formula as half the product of two adjacent side lengths with the sine of the angle formed by the same two sides. Thus, these triangle areas are $\frac12 r_1r_2\sin\frac{\pi}{3}$, $\frac12 r_2r_3\sin\frac{\pi}{3}$, and $\frac12 r_1r_3\sin\frac{2\pi}{3}$ respectively. But $\sin\pi/3=\sin 2\pi/3$, so the already-obtained equation $r_1r_2+r_2r_3=r_1r_3$ implies that triangles $C_1XC_2$ and $C_2XC_3$ together have the same total area as triangle $C_1XC_3$. This could only happen if the three centers $C_1$, $C_2$, and $C_3$ are collinear, for otherwise the sum of the areas of $C_1XC_2$ and $C_2XC_3$ would differ from the area of $C_1XC_3$ by the area of triangle $C_1C_2C_3$, which is zero only when these three points are collinear.

\item[{\rm (2)${}\Rightarrow{}$(1):}]~\\
By (2) the three center points $C_i$ are collinear; assume without loss of generality that $C_2$ lies between $C_1$ and $C_3$ on their common line. Because the three centers $C_1$, $C_2$, and $C_3$ form angles of $\pi/3$ rather than $2\pi/3$, the signed curvature of the middle center $C_2$ has the opposite sign to the signed curvature of the other three curvatures; we may assume without loss of generality (by mirror reversing the configuration if necessary) that these three signed curvatures are $1/r_1$, $-1/r_2$, and $1/r_3$. As in the previous case, the three lines $XC_i$ form angles of $\pi/3$ with each other so angles $C_1XC_2$ and $C_2XC_3$ must both equal $\pi/3$, angle $C_1XC_3$ must equal $2\pi/3$, and by collinearity the two triangles $C_1XC_2$ and $C_2XC_3$ together disjointly cover the same region of the plane as the single triangle $C_1XC_3$.

We can apply the side-angle-side formula to this region in two different ways, giving the equation
$$\frac12 r_1r_2\sin\frac{\pi}{3}+\frac12 r_2r_3\sin\frac{\pi}{3}=\frac12 r_1r_3\sin\frac{2\pi}{3}.$$
But the factors of $1/2$ cancel, as do the factors of $\sin\pi/3=\sin 2\pi/3$,
leaving the simpler equation $r_1r_2+r_2r_3=r_1r_3$.
Dividing all terms by $r_1r_2r_3$ gives $1/r_3+1/r_1=1/r_2$, and rearranging gives
$1/r_1-1/r_2+1/r_3=0$ as desired.

\item[{\rm (2)${}\Rightarrow{}$(3):}]~\\
The three circles have at least one triple crossing point, at $X$. Let $\ell$ be the line through the three centers, assumed to exist by (2). Then because $\ell$ passes through each circle center, a reflection across $\ell$ is a symmetry of each circle and therefore of the whole configuration of three circles. The point $X$ cannot lie on $\ell$, because two circles centered on a line cannot cross at a point that is also on the line, they can only meet at a point of tangency, violating the assumption that the arcs meet at angles of $2\pi/3$. Therefore, the reflection of $X$ across $\ell$ is also a triple crossing point, and the three circles have two triple crossings as (3) states.

\item[{\rm (3)${}\Rightarrow{}$(2):}]~\\
For any two intersecting circles, it is necessarily true that their centers lie on the perpendicular bisector of the chord connecting their two intersection points. But by the assumption of (3) that there are two triple crossing points, the chords defined by each pair of the three given circles coincide; therefore, their perpendicular bisectors also coincide in a line $\ell$ that contains all three circles.
\end{description}
\end{proof}

\end{document}